 \documentclass[final,5p,times,twocolumn]{elsarticle}

\usepackage{amssymb}
\usepackage{amsmath}
\usepackage{array}
\usepackage{color}
\usepackage{cuted}
\usepackage{arydshln}
\usepackage{multirow}  % 提供跨行合并功能
\newtheorem{theorem}{\textbf{Theorem}}
\newtheorem{lemma}{\textbf{Lemma}}

\newtheorem{corollary}{\textbf{Corollary}}
\newtheorem{remark}{\textbf{Remark}}
\newtheorem{definition}{\textbf{Definition}}

\newenvironment{proof}{{{\bf Proof:}}}{\hfill $\square$\par}
\journal{Elsevier}
\usepackage{setspace}  
\begin{document}
\begin{frontmatter}



%\title{Data-Driven Predictive Control in Industrial Descriptor Systems: Practical Frequency and Pressure Regulation for Power System and Water Networks}

\title{Data-Driven Analysis and Predictive Control of Descriptor Systems with Application to Power and Water Networks}
%\title{Data-Driven Analysis and Predictive Control of Descriptor Systems with Cases Studies in Power and Water Networks}
%\title{Data-Driven Controllability and Observability Analysis of Descriptor Systems with Application to Power and Water Networks}
%\title{Data-Driven Controllability and Observability Analysis with Application to Predictive Control of Descriptor Systems}

%\author{Yuan Zhang, Yu Wang, Yuanqing Xia} %% Author name

		%\thanks[footnoteinfo]{This work was supported in part by the
	%		National Natural Science Foundation of China under Grant 62373059.}
		
		{\small{	
        \author[bit]{Yuan Zhang}\ead{zhangyuan14@bit.edu.cn}    % Add the
	\author[bit]{Yu Wang}\ead{3220241197@bit.edu.cn}               % e-mail address
	\author[tj]{Jun Shang}\ead{shangjun@tongji.edu.cn}
	\author[bit]{Yuanqing Xia}\ead{xia\_yuanqing@bit.edu.cn} 
        \author[bit]{Jinhui Zhang*}\ead{zhangjinh@bit.edu.cn (corresponding author)}                      
                }} % (ead) as shown              

\affiliation[bit]{organization={School of Automation, Beijing Institute of Technology},%Department and Organization
            city={Beijing},
            postcode={100081}, 
            country={China }}
            
\affiliation[tj]{organization={Department of Control Science and Engineering, Shanghai Institute of Intelligent Science and Technology, State Key Laboratory of Autonomous Intelligent Unmanned Systems, and Frontiers Science Center for Intelligent Autonomous Systems, Tongji University},%Department and Organization
            city={Shanghai},
            postcode={200092}, 
            country={China}}
  \begin{abstract}  Despite growing interest in data-driven analysis and control of linear systems, descriptor systems—which are essential for modeling complex engineered systems with algebraic constraints like power and water networks—have received comparatively little attention. This paper develops a comprehensive data-driven framework for analyzing and controlling discrete-time descriptor systems without relying on explicit state-space models. We address fundamental challenges posed by non-causality through the construction of forward and backward data matrices, establishing data-based sufficient conditions for controllability and observability in terms of input-output data, where both R-controllability and C-controllability (R-observability and C-observability) have been considered. We then extend Willems' fundamental lemma to incompletely controllable systems. These methodological advances  Data-Enabled Predictive Control (DeePC) to achieve output tracking in descriptor systems and to maintain performance under incomplete controllability conditions, as demonstrated in two case studies: i) Frequency regulation in an IEEE 9-bus power system with $3$ generators, where DeePC maintained the frequency stability of the power system despite deliberate violations of R-controllability, and ii) Pressure head control in an EPANET water network with $3$ tanks, $2$ reservoirs, and $117$ pipes, where output tracking was successfully enforced under algebraic constraints.  
\end{abstract}

% \iffalse
% \fi 
% \begin{strip}
%     \begin{highlights}
%     \item Research highlight 1
    
%     Using only the input-output data matrix, sufficient conditions for judging the (R- and C-) controllability and (R- and C-) observability of descriptor systems are derived.  The forward and backward data matrices are constructed to deal with the non-causality of descriptor systems.
%     \item Research highlight 2

%     %\item Research highlight 3
    
%     We extend Willems' fundamental lemma for descriptor systems to incomplete controllable systems and reduce the order of persistency of excitation, leading to a paradigm of online
% Data-Enabled
% Predictive Control (DeePC)  for uncontrollable descriptor systems.
    
%     \item Research highlight 3
    
%     Case studies on typical power and water networks verify that online DeePC can ensure that the controllable output tracks the reference value even though the descriptor system contains uncontrollable output.
%     \end{highlights}
% \end{strip}

%% Keywords

\end{frontmatter}



\section{Introduction}
    Controllability and observability, as core properties in system and control theory, play a critical role in analyzing dynamic behavior and designing effective control strategies \cite{kalman1960contributions}. Traditional control theory primarily addresses linear systems, which offer a straightforward mathematical framework and a well-developed theoretical foundation \cite{mishra2020data}. These systems have found broad applications in fields like industrial control and signal processing. However, real-world engineering systems often exhibit complexities, such as algebraic constraints, time-delay coupling, and multi-modal interactions that lead to singular models. In contrast to standard linear systems, descriptor systems (or singular systems), especially discrete-time descriptor systems, leverage differential-algebraic equations to more comprehensively capture the hybrid dynamic characteristics of complex systems, demonstrating greater applicability in areas like power networks, traffic dispatching, chemical processing and manufacturing applications, and economic modeling \cite{Dai1989SingularCS,wu2013reliable,jafari2020observer}.		
		
    Although remarkable progress has been made for descriptor systems, most of the existing results rely on accurate model parameters.  In practice, modeling uncertainties and parameter perturbations often limit the effectiveness of traditional model-based methods. In contrast, data-driven approaches directly use input, output, or state data without system identification, which significantly reduces the dependence on prior model information and provides a new idea for the analysis and control of complex systems. For instance, \cite{he2021data} tackled the challenge of analyzing dynamic-algebraic hybrid structures through a data-driven controllability framework, using excitation input experiments to replace explicit model parameters with data matrices derived from input-state trajectories. Yet, the design of controllability and observability criteria for descriptor systems with incomplete state measurements remains unresolved.

     Recent years have witnessed a growing interest in data-driven control methodologies, with Willems' fundamental lemma \cite{willems2005note} serving as a cornerstone. This framework enables direct controller synthesis from experimental data, bypassing explicit system identification. Research advancements span diverse system classes: deterministic linear time-invariant (LTI) systems via non-parametric representations \cite{coulson2019data,tang2023multi}, linear parameter-varying  architectures \cite{verhoek2021fundamental,hamdan2024data}, stochastic LTI formulations \cite{pan2021stochastic}, and generalizations to nonlinear dynamics \cite{han2025big} encompassing polynomial \cite{strasser2020data} and non-polynomial structures \cite{alsalti2021data,putri2025stability}. Such data-driven paradigms facilitate behavioral system analysis \cite{markovsky2021behavioral}, controller design through subspace methods \cite{de2019formulas}, and predictive control implementations \cite{berberich2020data,o2021data,hamdan2024data,putri2025stability}. Given the prevalence of descriptor systems in modeling practical dynamic processes, Schmitz et al. \cite{schmitz2022willems} have extended this framework to linear discrete-time descriptor systems by proposing a tailored fundamental lemma variant. Their analysis reveals that regular descriptor systems require reduced Hankel matrix data dimensions compared to standard LTI cases while maintaining comparable input signal excitation conditions—a critical insight for practical implementations.
     However, their extension requires the descriptor systems to be both R-controllable and R-observable. 
     Currently, it remains unclear how to verify these conditions directly from data, and whether the requirements of R-controllability and persistence of excitation in Willems' fundamental lemma for descriptor systems are necessary. 		

    In response to the aforementioned issues, in this paper, we first undertake a systematic study on the controllability and observability analysis of discrete-time descriptor systems from a data-driven perspective. Based on the rank condition of the input-output data matrix, we establish several sufficient conditions for R-controllability, complete controllability (C-controllability), R-observability, and complete observability (C-observability) of descriptor systems. In particular, to address the non-causality-induced challenges in C-controllability analysis, some forward and backward data matrices are constructed, while only the forward data matrices are needed for non-singular systems. 
        
    Furthermore, we also extend Willems' fundamental lemma to incompletely controllable descriptor systems, eliminating the need for R-controllability in Schmitz et al.'s work  \cite{schmitz2022willems}.  Our approach builds on the data-driven paradigm in \cite{yu2021controllability}, but extends it to descriptor systems. We address this problem by answering the following question: to what extent can a linear combination of a finite number of measured input-output trajectories parameterize all possible trajectories without R-controllability. It is demonstrated that when the excitation is sufficiently persistent, any input-output trajectory of length $L$ with the initial state in a particular subspace can be expressed as a linear combination of a finite number of measurement trajectories. This subspace is composed of the Cartesian product of the sum of an R-controllable subspace, an unobservable subspace of the slow subsystem, and a certain subspace related to the measurement trajectory associated subspaces of the slow subsystem, with the controllable subspace of the fast subsystem. This naturally yields the development of online data-enabled predictive control (DeePC). By collecting a segment of online trajectory to parameterize all controllable trajectories, online DeePC applies to incompletely controllable descriptor systems.   We further use these results for frequency control in power systems and pressure head regulation in water networks. Specifically, we illustrate the performance of online DeePC on a 9-bus power system with $3$ generators and an EPANET water network with $3$ tanks, $2$ reservoirs, $2$ pumps, $92$ junctions, and $117$ pipes. 

    Notation: $\mathbb{R}$, $\mathbb{N}_0$ and $\mathbb{N}$ denote the set of real numbers, the natural numbers with and without zero. The identify matrix in $\mathbb{R}^{n \times n}$ and the zero matrix in $\mathbb{R}^{n \times n}$ are denoted by $I_n$ and $0_{n\times n}$, respectively. Further, for $k\in \mathbb{N}$ let $diag_L(A) = I_L \otimes A$, where $\otimes$ denote the Kronecker product. For a matrix $A\in \mathbb{R}^{m\times n}$, we denote by $rk(A)$, $Im(A)$ and $Ker(A)$ the rank, the image and the right kernel of $A$. When applied to subspaces, we let $+$ and $\times$ denote the sum and Cartesian product operation, respectively. Given a time series $w:[1,T]\rightarrow \mathbb{R}^q$, $w_{[t_1,t_2]}$ is the restriction of $w$ to the finite time interval $[t_1,t_2]\cap {\mathbb N}_0$.
        The associated Hankel matrix with $L \in \mathbb{N}$ block-rows is defined as
		$$
		W_{1,L,T - L + 1} = 
		\begin{bmatrix}
			\begin{smallmatrix}
				w(1) & w(2) & \cdots & w(T - L + 1)\\
				w(2) & w(3) & \cdots & w(T - L + 2)\\
				\vdots & \vdots & \ddots & \vdots\\
				w(L) & w(L + 1) & \cdots & w(T)
			\end{smallmatrix}
		\end{bmatrix}:= H_L(w).
		$$This time series is persistently exciting of order $L$ if the Hankel matrix $H_L(w)$ is full row rank.
\section{Preliminaries}
    We consider discrete-time linear descriptor systems
    \begin{subequations}
            \begin{align}
                Ex(k+1) &= Ax(k) + Bu(k) \label{2_1a} \\
                y(k)    &= Cx(k) + Du(k) \label{2_1b}
            \end{align}
            \label{2_1}
    \end{subequations}with (consistent) initial condition $Ex(0) = x^0$, where $A,E \in \mathbb{R}^{n\times n}$, $B \in \mathbb{R}^{n \times m}$, $C\in \mathbb{R}^{p\times n}$, $D\in \mathbb{R}^{p\times m}$. We assume that $\det(\lambda E - A) \ne 0$ for some $\lambda \in \mathbb{C}$, i.e., regularity of system (\ref{2_1a}). Particularly, we are interested in the case where the matrix $E$ is singular, i.e., $rk(E) < n$. When $k = 0, 1, \dots, L$, we view (\ref{2_1}) as a system of finite time series.		
		
    Since the descriptor system (\ref{2_1a}) is regular, there exist invertible matrices $P,S\in \mathbb{R}^{n\times n}$ such that (cf. \cite{Dai1989SingularCS} and \cite{berger2012quasi})
		$$
		SEP = 						
		\begin{bmatrix}
			I_{q} & 0 \\
			0 & N
		\end{bmatrix},
		SAP =
		\begin{bmatrix}
			A_{1} & 0 \\
			0 & I_{r}
		\end{bmatrix}
		$$
		$$
		SB =						
		\begin{bmatrix}
			B_{1} \\
			B_{2}
		\end{bmatrix},
		CP =
		\begin{bmatrix}
			C_{1} & C_{2}
		\end{bmatrix}
		$$
    where $N\in \mathbb{R}^{r\times r}$ is nilpotent with nilpotency index $s\in\mathbb{N}$, and $A_1\in \mathbb{R}^{q\times q}$, $B_1\in \mathbb{R}^{q\times m}$, $B_2\in \mathbb{R}^{r\times m}$, $C_1\in \mathbb{R}^{p\times q}$, $C_2\in \mathbb{R}^{p\times r}$ with $q+r=n$. Upon introduction of the coordinate change $z=P^{-1}x$, system (\ref{2_1}) can equivalently be written in quasi-Weierstra$\beta$ form, i.e.,		
		\begin{equation}
			\begin{aligned}
				\begin{bmatrix}
					I_{q} & 0 \\
					0 & N
				\end{bmatrix}z(k+1)
				= \begin{bmatrix}
					A_{1} & 0 \\
					0 & I_{r}
				\end{bmatrix}z(k)
				+
				\begin{bmatrix}
					B_{1} \\
					B_{2}
				\end{bmatrix}u(k),
			\end{aligned}
			\label{2_2}
		\end{equation}		
		\begin{equation}
			y(k)=\begin{bmatrix}
				C_{1} & C_{2}
			\end{bmatrix}z(k)+Du(k).
			\label{2_3}
		\end{equation}	
        
     Given an input trajectory $u:\mathbb{N}_{0}\to\mathbb{R}^{m}$ and an initial value $z_{1} ^{0}\in\mathbb{R}^{q}$ there is a unique trajectory such that the state
        $
    	z=
    		\begin{bmatrix}
    			z_{1}^{\top} &	z_{2}^{\top}          \end{bmatrix}^{\top}:\mathbb{N}_{0}\to\mathbb{R}^{q+r}
        $
    satisfies $z_{1}(0)=z_{1}^{0}$. This state $z(k)$, $k\in\mathbb{N}_{0}$, is given by		
	\begin{equation}
			z_{1}(k)=A_{1}^{k}z_{1}(0)+\sum_{m=1}^{k} A_{1}^{k-m}B_{1}u(m-1),
			\label{2_4}
		\end{equation}
            \begin{equation}
			z_{2}(k)=N^{L-k}z_2(L)-\sum_{m=0}^{L-k-1}N^{m}B_{2}u(k+m), k = 0, 1, \dots, L
			\label{2_5}
		\end{equation}
		or
		\begin{equation}
			z_{2}(k)=-\sum_{m=0}^{s-1}N^{m}B_{2}u(k+m), k = 0, 1, \dots, L, \dots.
			\label{2_6}
		\end{equation}
        
    Here, \eqref{2_4} captures the dynamics of the {\emph{slow subsystem} of system \eqref{2_1}, while \eqref{2_5} the {\emph{fast subsystem}}. Since in the fast subsystem, the current state depends on the future inputs,  discrete-time descriptor systems are generically non-causal. 
        
    Next, we review the concepts of R-controllability, R-observability, C-controllability, and C-observability; see \cite{Dai1989SingularCS}.

    For any fixed terminal condition $ z_2(L) \in \mathbb{R}^{r}$, introduce the reachable set $\mathcal{R}(z_2(L))$ for system (\ref{2_1}) with an arbitrary initial state, defined as follows, which will be simplied called the initial reachable set
        $$
        \mathcal{R}(z_2(L)) = \left\{ w \, \bigg| \, 
        \begin{aligned}
        & w \in \mathbb{R}^n, \ \exists\, z_1(0),\ 0 \leq k_1 \leq L, \\
        & u(0), u(1), \ldots, u(L) \text{ such that } x(k_1) = w
        \end{aligned}
        \right\}.
        $$
    Evidently, the initial reachable set $\mathcal{R}(z_2(L))$ depends on $z_2(L) $. For different terminal conditions $ z_2(L) $, the reachable sets $\mathcal{R}(z_2(L))$ may differ.
    \begin{definition}[R-controllability \cite{Dai1989SingularCS}]
        Finite time series (\ref{2_1}) is called controllable in the initial reachable set (R-controllable) if for any fixed terminal condition $z_2(L)$, the state starting from any initial condition can be controlled by inputs to reach any state in $\mathcal{R}(z_2(L))$ at a certain time point.
    \end{definition}
    \begin{definition}[R-observability \cite{Dai1989SingularCS}]
        Finite time series (\ref{2_1}) is called R-observable if it is observable in the initial reachable set $\mathcal{R}(z_2(L))$ for any fixed terminal condition $z_2(L)\in \mathbb{R}^r$.
    \end{definition}
    \begin{lemma}[\cite{belov2018control}]
    System (\ref{2_1}) is R-controllable if and only if
		\begin{equation} \label{2_10}
				rk\begin{bmatrix}
					B_1 & A_1B_1 & \dots & A_1^{q-1}B_1
				\end{bmatrix} = q,
		\end{equation}
    and is R-observable if and only if
		\begin{equation} \label{2_11}
				rk \begin{bmatrix}
					C_1\\
					C_1A_1\\
					\vdots\\
					C_1A_1^{q-1}
				\end{bmatrix}  = q.
		\end{equation}
    \end{lemma}
    \begin{definition}[C-controllability \cite{Dai1989SingularCS}]
        Finite time series (\ref{2_1}) is called C-controllable if for any complete condition $\begin{bmatrix}
    	z_1(0)\\ z_2(L)
        \end{bmatrix}$ and $w\in\mathbb{R}^n$ there exists a time point $k_1$, $0\le k_1\le L$, and control inputs $u(0),u(1),\dots,u(L)$ such that $x(k_1) = w$.
    \end{definition}
		\begin{definition}[C-observability \cite{Dai1989SingularCS}]

    A finite time series (\ref{2_1}) is called C-observable if its state $x(k)$ at any time point k is uniquely determined by $\{u(i),y(i), i = 0,1,\dots, L\}$.
    \end{definition}
    \begin{lemma}[\cite{belov2018control}]
    The descriptor system (\ref{2_1}) is \textit{C-controllable} if and only if
		\begin{subequations} \label{2_8}
			\begin{align}
				\text{rank} \begin{bmatrix}
					B_1 & A_1B_1 & \dots & A_1^{q-1}B_1
				\end{bmatrix} = q, \label{2_8a} \\
				\text{rank} \begin{bmatrix}
					B_2 & NB_2   & \dots & N^{r-1}B_2
				\end{bmatrix} = r. \label{2_8b}
				\end{align}
		\end{subequations}
    and is \textit{C-observable} if and only if
		\begin{subequations} \label{2_9}
			\begin{align}
				\text{rank} \begin{bmatrix}
					C_1       \\
					C_1A_1    \\
					\vdots    \\
					C_1A_1^{q-1}
			\end{bmatrix} = q, \label{2_9a} \\
			\text{rank} \begin{bmatrix}
				    C_2     \\
					C_2N    \\
					\vdots  \\
					C_2N^{r-1}
				\end{bmatrix} = r. \label{2_9b}
			\end{align}
		\end{subequations}
    \end{lemma}	
    \begin{remark}[Controllability conditions]
    The Kalman controllability rank condition in (\ref{2_8a}) is equivalent to 
	\begin{equation}\label{2_12}
			rk\begin{bmatrix}
				A_1-\lambda I & B_1
			\end{bmatrix}=q, \,\ \forall \lambda \in \mathbb{C},
	\end{equation}
    (\ref{2_8b}) is equivalent to 
	\begin{equation}\label{2_13}
		rk\begin{bmatrix}
			N-\lambda I & B_2
		\end{bmatrix}=r, \,\ \forall \lambda \in \mathbb{C}.
	\end{equation}
    This demonstrates the connection between Kalman's rank condition and the PBH criterion.
		\end{remark}

%{\color{blue}In our subsequent data-based analysis, the parameters $(A,B,C,D)$ are not known. However, it is assumed that the system dimension $n$, the nilpotency index $s$, and the dimensions of the fast subsystem $r$ and the slow subsystem $q$ are known apriori. These parameters are necessary to conduct the subsequent analysis, since we only have access to the system input-output data, consistent with which descriptor state-space realizations with different state dimensions can exist.}

{In our subsequent data-driven analysis, the descriptor state-space parameters $(E, A, B, C, D)$ remain unknown. However, we {assume prior knowledge of} the system dimension $n$, the nilpotency index $s$, and the dimensions of the fast and slow subsystems ($r$ and $q$, respectively). These parameters {are essential for our analysis}, as we {rely solely on input-output data}. {Critically}, multiple descriptor state-space realizations {with differing state dimensions} can {align with} the same input-output data. }
        
\section{Data-driven Controllability Tests}
    Conventional approaches for the controllability analysis of descriptor systems typically involve two steps: first identifying the system model, then performing matrix computations based on model-based criteria. This two-step process can be computationally intensive. This section establishes several sufficient conditions in terms of input-output data for the controllability of descriptor systems.
    
   % (i) Algebraic constraints induce dimensional overlapping in state-space decomposition; (ii) Parameter uncertainties propagate computational errors in controllable subspace determination. This section establishes several sufficient conditions in terms of input-output data for the controllability of descriptor systems.

    The set of infinite input-output trajectories of system \eqref{2_1} is given as 
        $$
        {\cal B} = \left\{ (u,y): {\mathbb{N}}_0\to {\mathbb R}^m\times {\mathbb R}^{p} \, \bigg| \, 
        \begin{aligned}
        \exists x: \mathbb{N}_0\to {\mathbb R}^n: x(t), u(t), y(t)\\
        \ {\rm  satisfy} \ \eqref{2_1}\ {\rm for} \ {\rm all}\  t\in {\mathbb N}_0
        \end{aligned}
        \right\}.
        $$
    Throughout this paper, by an input-output trajectory $(U,Y)$ of system \eqref{2_1}, we mean there is an infinite trajectory $(u,y)\in {\cal B}$ such that $U=u_{[t,T]}$ and $Y=y_{[t,T]}$ for some finite time interal $[t,T]\cap {\mathbb N}_0$, i.e., $(U,Y)$ is the restriction of $(u,y)$ to the time interval $\{t,t+1,\cdots,T\}$.
\begin{theorem}\label{thm:R-cont}
    Suppose system (\ref{2_1}) is regular. Let $(U, Y)$ be an input-output trajectory of system (\ref{2_1}). Then, system (\ref{2_1}) is R-controllable if there exists $L \in \mathbb{N}$ such that
        \begin{equation}
		rk\left(\left[\begin{array}{l}
			U_{1, L, T-L+1} \\
			Y_{1, L, T-L+1}
		\end{array}\right]-\lambda\left[\begin{array}{l}
			U_{0, L, T-L+1} \\
			Y_{0, L, T-L+1}
		\end{array}\right]\right)=m(L+s-1)+q
		\label{3_1}
	\end{equation} holds for all $\lambda \in \mathbb{C}$.
\end{theorem}
\begin{proof}	
    Since the I-O data is invariant with respect to invertible state transformation, assume that system (\ref{2_1}) is given in the quasi-Weierstra$\beta$ form.

    Define 
	\begin{align*}
		O_{sL}&=\begin{bmatrix}
			C_1\\
			C_1A_1\\
			\vdots\\
			C_1A_1^{L-1}
		\end{bmatrix},\\
        T_L&=\begin{bmatrix}
			D - C_2B_2 & -C_2NB_2 & \cdots & 0\\
			\vdots & \vdots & \ddots & \vdots\\
			C_1A_1^{L - 2}B_1 & C_1A_1^{L - 3}B_1 & \cdots & -C_{2}N^{s - 1}B_{2}\\
		\end{bmatrix}.
	\end{align*}   
    
    Note that,
        \begin{align*}
            &\begin{bmatrix}
                H_L(u_{[1,T]})\\
                H_L(y_{[1,T]})
            \end{bmatrix}-\lambda\begin{bmatrix}
                H_L(u_{[0,T - 1]})\\
                H_L(y_{[0,T - 1]})
            \end{bmatrix}\\
            &=\left[\begin{array}{c:c}
                \!0\! & [I_{mL}\,\, 0]\\
                \! O_{sL} \! &  T_L
            \end{array}\right]\begin{bmatrix}
                \! H_1(z_{1[1,T - L + 1]})-\lambda H_1(z_{1[0,T - L]})\!\\
                \! H_{L + s - 1}(u_{[1,T + s - 1]})-\lambda H_{L + s - 1}(u_{[0,T + s - 2]})\!
            \end{bmatrix}.
        \end{align*}
        Therefore,
	\begin{align*}
		&rk\left(\left[\begin{array}{l}H_L(u_{[1, T]}) \\ H_L(y_{[1,T]})\end{array}\right]-\lambda\left[\begin{array}{c}H_L\left(u_{[0, T-1]}\right) \\ H_L\left(y_{[0, T-1]}\right)\end{array}\right]\right) \leq \\ &rk\left(H_1(z_{1[1, T-L+1]})-\lambda H_1\left(z_{1[0,T-L]}\right)\right)+m(L+s-1).
	\end{align*}	    
	To make condition (\ref{3_1}) hold, it must have
	\begin{equation}
		rk\left(H_1(z_{1[1, T-L+1]})-\lambda H_1\left(z_{1[0,T-L]}\right)\right)=q.
		\label{3_2}
	\end{equation}	
	From equation (\ref{2_4}),
        \begin{equation} \label{reason}
	\begin{aligned}
		H_1\left(z_{1[1, T-L+1]}\right)- \lambda H_1\left(z_{1[0, T-L]}\right)= 
		\begin{bmatrix}
			A_1-\lambda I & B_1
		\end{bmatrix}
		\begin{bmatrix}
			H_1(z_{1[0,T-L]})\\
			H_1(u_{[0,T-L]})
		\end{bmatrix}.
	\end{aligned}
        \end{equation}
	Therefore, if (\ref{3_2}) holds, then	
	$$
	rk\begin{bmatrix}
		A_1-\lambda I & B_1
	\end{bmatrix}=q, \,\ \forall \lambda \in \mathbb{C}
	$$	
	which means system (\ref{2_1}) is R-controllable.
\end{proof}
\begin{theorem}\label{thm:cont}
    Suppose system (\ref{2_1}) is regular. Let
	$(U, Y)$ be an input-output trajectory of system (\ref{2_1}). Then, system (\ref{2_1}) is C-controllable if there exists $L \in \mathbb{N}$ such that
	\begin{equation}
		r k\left(\left[\begin{array}{l}
			U_{1, L, T-L+1} \\
			Y_{1, L, T-L+1}
		\end{array}\right]-\lambda\left[\begin{array}{l}
			U_{0, L, T-L+1} \\
			Y_{0, L, T-L+1}
		\end{array}\right]\right)=mL+n
		\label{3_3}
	\end{equation}
	and
	\begin{equation}
		r k\left(\left[\begin{array}{l}
			U_{0, L, T-L+1} \\
			Y_{0, L, T-L+1}
		\end{array}\right]-\lambda\left[\begin{array}{l}
			U_{1, L, T-L+1} \\
			Y_{1, L, T-L+1}
		\end{array}\right]\right)=mL+n
		\label{3_4}
	\end{equation}
	for all $\lambda \in \mathbb{C}$.
\end{theorem}
\begin{proof}
%	The PBH criterion for controllability of descriptor systems requires:
%	\begin{align*}
%		rk\begin{bmatrix}
%			\lambda I - A_1 & B_1
%		\end{bmatrix} = q, \quad \forall \lambda \in \mathbb{C}\\
%		rk\begin{bmatrix}
%			\lambda I - N & B_2
%		\end{bmatrix} = r,\quad \forall \lambda \in \mathbb{C}
%	\end{align*}
%	By employing the finite-time series model of the fast subsystem defined in equation (\ref{2_5}), we establish the mathematical formulation presented in equation (\ref{2_10}) and (\ref{2_11}).
Again, we resort to the quasi-Weierstra$\beta$ form of system (\ref{2_1}).	
	Define \[
         S=\begin{bmatrix}
		I_{q} \\
		A_{1} \\
		\vdots \\
		A_1^{L-1}
	\end{bmatrix},
	\mathcal{T}=\begin{bmatrix}
		0 & \cdots & 0 & 0 \\
		B_{1} & \cdots & 0 & 0 \\
		\vdots & \ddots & \vdots & \vdots \\
		A_1^{L-2}B_{1} & \cdots & B_{1} & 0
	\end{bmatrix},
	U = \begin{bmatrix}
		N^{L-1} \\ N^{L-2} \\ \vdots \\ I_r
	\end{bmatrix},
	\]
	\[	
    O_{fL}=\begin{bmatrix}
		C_2N^{L-1}\\
		C_2N^{L-2}\\
		\vdots\\
		C_2
	\end{bmatrix}, V = \begin{bmatrix}
		B_2 & NB_2 & \dots & N^{L-2}B_2 & 0\\
		0 & B_2 & \dots & N^{L-3}B_2 & 0\\
		\vdots & \vdots & \ddots & \vdots & \vdots\\
		0 & 0 & \dots & B_2 & 0\\
		0 & 0 & \dots & 0 & 0
	\end{bmatrix}.
    \]    
    We have that
	\begin{equation}
		\begin{bmatrix}
			z_{1}[0,L-1]\\z_{2}[0,L-1]\\u[0,L-1]
		\end{bmatrix}=\left[\begin{array}{c:c:c}
			S & 0 & \mathcal{T} \\
			0 & U & -V \\
			0 & 0 & I
		\end{array}\right]\begin{bmatrix}
		z_{1}(0)\\
		z_{2}(L-1)\\
		u_{[0,L-1]}
		\end{bmatrix} \label{3_5}
	\end{equation} 
	and
	\begin{equation}
		\begin{bmatrix}
			u_{[0,L-1]}\\y_{[0,L-1]}
		\end{bmatrix}=
		\begin{bmatrix}
			0 & 0 & I\\
			diag_{L}(C_1) & diag_{L}(C_2) & diag_{L}(D)
		\end{bmatrix}
		\begin{bmatrix}
			z_{1}[0,L-1]\\z_{2}[0,L-1]\\u[0,L-1]
		\end{bmatrix}.\label{3_6}
	\end{equation}	
    Substituting  (\ref{3_5}) into (\ref{3_6}) yields 
	\begin{equation}
		\begin{bmatrix}
			u_{[0,L-1]}\\y_{[0,L-1]}
		\end{bmatrix}
		=Q\begin{bmatrix}
			z_{1}(0)\\
			z_{2}(L-1)\\
			u_{[0,L-1]}
		\end{bmatrix},
	\end{equation}
    where 	
	\[
	Q = \begin{bmatrix}
			0 & 0 & I \\
			O_{sL} & O_{fL} & diag_{L}(C_1)\
				\mathcal{T}-diag_{L}(C_2)V+diag_{L}(D)
	\end{bmatrix}.
	\]	
    Thus, it's obvious that
    $$
          \begin{bmatrix}
			H_L(u_{[1,T]})\\
			H_L(y_{[1,T]})
		\end{bmatrix}-\lambda\begin{bmatrix}
			H_L(u_{[0,T-1]})\\
			H_L(y_{[0,T-1]}
		\end{bmatrix} = Q\begin{bmatrix}
			\begin{smallmatrix}
			H_1(z_{1[1,T-L+1]}) - \lambda H_1(z_{1[0,T-L]})\\
			H_1(z_{2[L,T]}) - \lambda H_1(z_{2[L-1,T-1]})\\
			H_{L}(u_{[1,T]})-\lambda H_{L}(u_{[0,T-1]})
			\end{smallmatrix}
		\end{bmatrix}.
    $$
    Therefore, 
	{
        \footnotesize
        \begin{align*}
            &rk\left( 
    		\begin{bmatrix} 
    			H_L(u_{[1, T]}) \\ 
    			H_L(y_{[1,T]})
    		\end{bmatrix} 
    		- \lambda 
    		\begin{bmatrix} 
    			H_L(u_{[0, T-1]}) \\ 
    			H_L(y_{[0, T-1]})
    		\end{bmatrix} 
    		\right) \leq 
    		rk\Big( H_1(z_{1[1, T-L+1]}) - \lambda H_1(z_{1[0,T-L]}) \Big) \\ &+ 
    		rk\Big( H_1(z_{2[L,T]}) - \lambda H_1(z_{2[L-1,T-1]}) \Big) + mL.
    	\end{align*}
        }
        Similarly, we have
        \begin{align*}
		& 
		\begin{bmatrix} 
			H_L(u_{[0, T-1]}) \\ 
			H_L(y_{[0,T-1]})
		\end{bmatrix} 
		- \lambda 
		\begin{bmatrix} 
			H_L(u_{[1, T]}) \\ 
			H_L(y_{[1, T]})
		\end{bmatrix} 
		   = Q\begin{bmatrix}
			\begin{smallmatrix}
			H_1(z_{1[0,T-L]}) - \lambda H_1(z_{1[1,T-L+1]})\\
			H_1(z_{2[L-1,T-1]}) - \lambda H_1(z_{2[L,T]})\\
			H_{L}(u_{[0,T-1]})-\lambda H_{L}(u_{[1,T]})
			\end{smallmatrix}
		\end{bmatrix}.
	\end{align*}
	Consequently, the rank constraint satisfies 
        {
        \footnotesize
        \begin{align*}
    		&rk\left( 
    		\begin{bmatrix} 
    			H_L(u_{[0, T-1]}) \\ 
    			H_L(y_{[0,T-1]})
    		\end{bmatrix} 
    		- \lambda 
    		\begin{bmatrix} 
    			H_L(u_{[1, T]}) \\ 
    			H_L(y_{[1, T]})
    		\end{bmatrix} 
    		\right) \leq 
    		rk\Big( H_1(z_{1[0, T-L]}) - \lambda H_1(z_{1[1,T-L+1]}) \Big) \\ &+
    		rk\Big( H_1(z_{2[L-1,T-1]}) - \lambda H_1(z_{2[L,T]}) \Big)  + mL.
    	\end{align*}
        }
    To make conditions (\ref{3_3}) and (\ref{3_4}) hold, it must have 
	$$rk\left(H_1(z_{1[1, T-L+1]})-\lambda H_1\left(z_{1[0,T-L]}\right)\right) = q$$
	and 
	$$
	rk\left(H_1(z_{2[L-1,T-1]}) - \lambda H_1(z_{2[L,T]})\right]=r.
	$$
    Similar to the reasoning in \eqref{reason}, we can get
	$$
	rk\begin{bmatrix}
		A_1-\lambda I & B_1
	\end{bmatrix} = q, \,\ \forall \lambda \in \mathbb{C}
	$$
	$$
	rk\begin{bmatrix}
		\lambda I-N & B_2
	\end{bmatrix}=r,\,\ \forall \lambda \in \mathbb{C}.$$
\end{proof}
\begin{remark}
    In view of \eqref{reason}, condition \eqref{3_3} can be seen as a forward data condition, corresponding to that the current state of the slow subsystem depends on the past state and input. In contrast, condition \eqref{3_4} is a backward data condition, corresponding to that the current state of the fast subsystem depends on the future state and input. 
    Since C-controllability concerns both the controllability of slow and fast subsystems, we need to construct both the forward data matrix and the backward data matrix. 
\end{remark}
\begin{remark}\label{rem:con}
    When the matrix $E$ is set to the identity matrix $(E = I)$, the descriptor system reduces to a standard linear time-invariant (LTI) system, and the fast subsystem vanishes entirely. Then, system (\ref{2_1}) is controllable if there exists $L \in \mathbb{N}$ such that
	\begin{equation*}
		rk\left(\left[\begin{array}{l}
			U_{1, L, T-L+1} \\
			Y_{1, L, T-L+1}
		\end{array}\right]-\lambda\left[\begin{array}{l}
			U_{0, L, T-L+1} \\
			Y_{0, L, T-L+1}
		\end{array}\right]\right) = n + mL.
	\end{equation*}
    Under the additional constraints of a persistently exciting input and a window length $L\ge n$, the sufficient controllability condition in Remark \ref{rem:con} is rigorously strengthened to a necessary and sufficient condition, achieving full equivalence with the data-driven characterization in Theorem 3 of \cite{mishra2020data}. 
\end{remark}
\section{Data-driven Observability Tests}
	%如定理3.1-3.2所示，数据驱动的能控性判据成功规避了传统方法对系统矩阵的依赖。作为对偶问题，能观性分析需要新的数据驱动表达方式。本节首先建立描述系统R-能观性的基础判据（定理4.1），进而推广至完全能观性场景（定理4.2）。
    As demonstrated by Theorems \ref{thm:R-cont} and \ref{thm:cont}, the data-driven controllability criteria successfully circumvent the dependence on system identification inherent in conventional methods. As a dual problem, observability analysis requires novel data-driven representations. This section first establishes a primary criterion for R-observability in descriptor systems, then extends the framework to complete observability scenarios.	
\begin{theorem}\label{thm:R-obser}
	Suppose system (\ref{2_1}) is regular. Let
	$(U, Y)$ be an input-output trajectory of system (\ref{2_1}). System (\ref{2_1}) is R-observable if there exists $L\ge q$ with $L\in \mathbb{N}$ such that		
	\begin{equation}
		rk\left[\begin{array}{l}
			U_{0, L, T-L+1} \\
			Y_{0, L, T-L+1}
		\end{array}\right]=m(L+s-1)+q.
		\label{4_1}
	\end{equation}
\end{theorem}
\begin{proof}
    Note that
	$$
	\begin{aligned}
		\begin{bmatrix}
			H_L(u_{[0,T-1]})\\
			H_L(y_{[0,T-1]})
		\end{bmatrix}
		= \left[\begin{array}{c:c}
			0 & [I_{mL}\,\, 0]\\
			O_{sL} & T_L
		\end{array}\right]\begin{bmatrix}
			H_1(z_{1[0,T-L]})\\
			H_{L+s-1}(u_{[0,T+s-2]})
		\end{bmatrix}.
	\end{aligned}
	$$	
	Note that $rk(\left[\begin{array}{c:c}
		0 & [I_{mL}\,\, 0]\\
		O_{sL} & T_L
	\end{array}\right]) \leq m(L+s-1)+q$. In order to make $\begin{bmatrix}
		H_L(u_{[0,T-1]})\\
		H_L(y_{[0,T-1]})
	\end{bmatrix}$ have rank $m(L+s-1)+q$, it must hold that
	$$
	rk\left[\begin{array}{c:c}
		0 & [I_{mL}\,\, 0]\\
		O_{sL} & T_L
	\end{array}\right] = m(L+s-1)+q
	$$		
	which requires that $rk(O_{sL})=q$, i.e., $\left(A_1, C_1\right)$ is observable.
\end{proof}
\begin{remark}
   For a standard LTI system, under the condition that the system is controllable and the input $U$ is persistently exciting of order $n+L$, condition \eqref{4_1} also becomes necessary for the observability \cite{verhaegen2007filtering}. However, for descriptor systems, establishing a necessary and sufficient condition for R-observability is challenging because even when $O_{sL}$ has full column rank, the matrix $\left[\begin{array}{c:c}
		0 & [I_{mL}\,\, 0]\\
		O_{sL} & T_L
	\end{array}\right]$  may be of column rank deficient (see the proof of Theorem \ref{thm:R-obser}). 
\end{remark}
\begin{theorem}\label{thm:obser}
    Suppose that system (\ref{2_1}) is regular. Let
	$(U, Y)$ be an input-output trajectory of system (\ref{2_1}). System (\ref{2_1}) is C-observable if $L\ge \max\{q,r\}$ with $L\in \mathbb{N}$ such that		
	\begin{equation}
		rk\left[\begin{array}{l}
			U_{0, L, T-L+1} \\
			Y_{0, L, T-L+1}
		\end{array}\right] = n+mL.
		\label{4_2}
	\end{equation}
\end{theorem}
\begin{proof}
    Note that
	$$
	\begin{aligned}
		\begin{bmatrix}
			H_L(u_{[0,T-1]})\\
			H_L(y_{[0,T-1]})
		\end{bmatrix}
		= Q\begin{bmatrix}
			H_1(z_{1[0,T-L]})\\
			H_1(z_{2[L-1,T-1]})\\
			H_{L}(u_{[0,T-1]})
		\end{bmatrix}.
	\end{aligned}
	$$	
    In order to make $\begin{bmatrix}
		H_L(u_{[0,T-1]})\\
		H_L(y_{[0,T-1]})
	\end{bmatrix}$ have rank $n+mL$, it must hold that
        \begin{align*}
		rk \left( \begin{bmatrix}
			\begin{smallmatrix}
			    0 & 0 & I \\
			O_{sL} & O_{fL} & \mathrm{diag}_{L}(C_1)\mathcal{T} - \mathrm{diag}_{L}(C_2)V + \mathrm{diag}_{L}(D)
			\end{smallmatrix}
		\end{bmatrix} \right)
		= n + mL
	\end{align*}		
	which requires that $rk(O_{sL})=q$ and $rk(O_{fL})=r$, i.e., $\left(A_1, C_1\right)$ and $(N, C_2)$ is observable.
\end{proof}
\begin{remark}
    As a corollary of Theorem \ref{thm:obser}, when the matrix $E$ is set to the identity matrix $(E = I)$, the descriptor system reduces to a standard LTI system. Then, system (\ref{2_1}) is observable if $L\ge n$ with $L\in \mathbb{N}$ such that	
    	\begin{equation*}
    		rk\left[\begin{array}{l}
    			U_{0, L, T-L+1} \\
    			Y_{0, L, T-L+1}
    		\end{array}\right] = n+mL. 
    	\end{equation*}
\end{remark}
\section{Extensions of Willems' fundamental lemma to uncontrollable systems} 
    While Willems' fundamental lemma for descriptor systems has been established under the R-controllability assumption in \cite{schmitz2022willems}, this prerequisite may limit its applicability in practical scenarios where the controllability structure is incomplete—rendering the original theorem inapplicable. To bridge this gap, inspired by \cite{yu2021controllability}, this section develops a generalized fundamental lemma by constructing a controllability-independent data-driven framework that reconstructs the system's behavioral characterization. This extension eliminates the need for prior controllability verification, as formalized in the subsequent Theorem \ref{thm:Willems}.

    We will introduce extensions of Willems' fundamental lemma and utilize the following subspaces:
    $$
    \cal{R} = \mathit{Im} \begin{bmatrix}
    	B_1 & A_1B_1 & \dots & A_1^{q-1}B_1
    \end{bmatrix},
    $$
    $$
    \cal{O}=\mathit{Ker}\begin{bmatrix}
    	C_1\\
    	C_1A_1\\
    	\vdots\\
    	C_1A_1^{q-1}
    \end{bmatrix},
    $$
    $$
    \mathcal{N}=\mathit{Im}\begin{bmatrix}
    	B_2 & NB_2 & N^2B_2 & \dots & N^{s-1}B_2
    \end{bmatrix},
    $$
    $$
    \cal{K}\begin{bmatrix}
    	z_{10}
    \end{bmatrix}=\mathit{Im}\begin{bmatrix}
    	z_{10} & A_1z_{10} & \dots & A_1^{q-1}z_{10}
    \end{bmatrix}.
    $$
    We term $\mathcal{N}$ the controllable subspace of the fast subsystem. Using the Cayley-Hamilton theorem, one can verify that (\ref{2_1}) ensures
    \begin{equation}
    	z_1(t)=A_1^tz_1(0)+\sum_{j=0}^{t-1}A_1^{t-j-1}B_1u_j  \in (\cal{R}+	\cal{K}\begin{bmatrix}
    		z_{10}
    	\end{bmatrix}).
    	\label{5_1}
    \end{equation}
    We let $\delta \in \mathbb{N}$ be such that $\delta \ge \delta_{\text{min}} :=degree\quad of\quad  p_{\text{min}}(A_1)$, where $p_{\text{min}}(A_1)$ is the minimal polynomial of matrix $A_1$ \cite{horn2012matrix}, and, due to the Cayley-Hamilton theorem, $\delta_{\text{min}}$ is upper bounded by $q$.
\begin{definition}
    We say a length-L input-output trajectory $(\bar{u}_{[0,L-1]},\bar{y}_{[0,L-1]})$, with $L\in\mathbb{N}$, is parameterizable by ${u_{[0,T-1]},y_{[0,T-1]}}$ if there exists $g\in \mathbb{R}^{T-L+1}$ such that
	\begin{equation}
		\begin{bmatrix}
			\bar{u}_{[0,L-1]}\\
			\bar{y}_{[0,L-1]}
		\end{bmatrix}=\begin{bmatrix}
			H_L(u_{[0,T-1]})\\
			H_L(y_{[0,T-1]})
		\end{bmatrix}g.
		\label{5_2}
	\end{equation}
\end{definition}
\begin{theorem}\label{thm:Willems}
	Suppose that the system (\ref{2_1}) is regular. Let $\{u_{[0,T-1]},z_{{[0,T-1]}},y_{[0,T-1]}\}$ be an input-state-output trajectory generated by the system (\ref{2_1}). If ${u_{[0,T-1]}}$ is persistently exciting of order $(L+\delta+s-1)$ where $\delta$ satisfies $\delta \ge \delta_{\text{min}} :=degree\quad of\quad  p_{\text{min}}(A_1)$, then
	\begin{equation}
		Im\begin{bmatrix}
			H_1(\bar{z}_{[0,T-L-s+1]})\\
			H_{L+s-1}(\bar{u}_{[0,T-1]})
		\end{bmatrix}=(\cal{R}+\cal{K}\begin{bmatrix}
			z_{10}
		\end{bmatrix})\times \mathcal{N}\times \mathbb{R}^{\text{m(L+s-1)}}.
		\label{5_3}
	\end{equation}
    Further, $(\bar{u}_{[0,L-1]},\bar{y}_{[0,L-1]})$ is parameterizable by $\{u_{[0,T-s]},y_{[0,T-s]}\}$ if and only if there exists a state trajectory $\bar{z}_{[0,L-1]}$ with
	\begin{equation}
		\bar{z}_{0}\in (\cal{R}+\cal{O}+\cal{K}\begin{bmatrix}
			z_{10}
		\end{bmatrix})\times \mathcal{N}
		\label{5_4}
	\end{equation}such that $\{\bar{u}_{[0,L-1]},\bar{z}_{{[0,L-1]}},\bar{z}_{[0,L-1]}\}$ is an input-state-output trajectory generated by (\ref{2_1}).	
\end{theorem}
\begin{proof} See the appendix. 
\end{proof}

    Due to its dependence on the state of $\bar{z}_{0}$ and the unknown subspace.
    The condition in $\bar{z}_{0}\in (\cal{R}+\cal{O}+\cal{K}\begin{bmatrix}
    	z_{10}
    \end{bmatrix})\times \mathcal{N}$ is difficult to verify, especially when only input-output traces are available. However, the following corollary provides an alternative that requires neither controllability nor state measurement. Through specific settings, the conditions of Theorem \ref{thm:Willems} are transformed into a more easily applicable form, so that conclusions about trajectory parameterization can be obtained without directly referring to the initial state of the system.

\begin{corollary}\label{cor:Willems}
    Let $(u_{[0,K-1]},y_{[0,K-1]})$ be an input-output trajectory generated by system (\ref{2_1}), and $u_{[0,T-1]}$ be persistently exciting of order $(L+\delta+s-1)$, where $\delta$ satisfies $\delta \ge \delta_{\text{min}}$ and $L+s-1\le T\le K+s-1$. Then $(u_{[t,t+L-1]},y_{[t,t+L-1]})$ is parameterizable by $(u_{[0,T-s]},y_{[0,T-s]})$ for $0\le t\le K-L$. Conversely, if $(\bar{u}_{[0,L-1]},\bar{y}_{[0,L-1]})$ is parameterizable by $(u_{[0,T-s]},y_{[0,T-s]})$, then it is an input-output trajectory generated by the system (\ref{2_1}).
\end{corollary}
\begin{proof}
    First, let $z_{[K-1]}$ be such that $(u_{[0,K-1]},z_{[0,K-1]},y_{[0,K-1]})$ is an input-state-output trajectory generated by the system (\ref{2_1}). Then (\ref{5_1}) holds and $z_2(k)\in \mathcal{N}$ holds for all $0\le t\le K-1$. Hence $(u_{[t,t+L-1]},z_{[t,t+L-1]},y_{[t,t+L-1]})$ is an input-state-output trajectory generated by the system (\ref{2_1}). From the second statement in Theorem \ref{thm:Willems} we know that $(u_{[t,t+L-1]},y_{[t,t+L-1]})$ is parameterizable by $(u_{[0,T-s]},y_{[0,T-s]})$. Second, if $(\bar{u}_{[0,L-1]},\bar{y}_{[0,L-1]})$ is parameterizable by $(u_{[0,T-s]},y_{[0,T-s]})$, then one can verify that it is indeed an input-output trajectory generated by the system (\ref{2_1}) by constructing a length-L state trajectory similar to the one in (\ref{5_14}).
\end{proof}

    Corollary 1 establishes trajectory parameterization results by introducing a specific configuration that reformulates the conditions of Theorem \ref{thm:Willems} into an applicable framework, thereby eliminating explicit dependence on the system’s initial states. It shows that any segment of length $L$ of the input-output trajectory can be parameterized by its first segment of length $(T-s)$, assuming sufficient persistence of the excitation. This theoretical result directly underpins the model-free predictive control framework of DeePC.
    % \begin{proposition}
    %     % Let system (\ref{2_1}) be R-observable and suppose that Q and R are symmetric positive definite. Let the horizon \textcolor{red}{$XX$}. According to Corollary \ref{cor:Willems}, even if the system (\ref{2_1}) is not R-controllable, for any trajectory $(u_{[0,K-1]},y_{[0,K-1]})$ generated by system (\ref{2_1}), the initial state of its arbitrary subsequence $(u_{[t,t+L-1]},y_{[t,t+L-1]})$ satisfies equation (\ref{5_4}). Provided that $u$ is persistently exciting of $\textcolor{red}{L+2(\delta+s-1)}$, the second statement of Theorem \ref{thm:Willems} holds. Consequently, if it is feasible at the initial time (that is, (\ref{6_1}) is feasible at time $t=T-s+1$), then feasibility is guaranteed for all $t\in\mathbb{N}$. Furthermore, for the predictive control closed-loop system, the target state, i.e., $(0, y^s)$, is globally asymptotically stable, where the attraction domain is implicitly characterized by the set of all feasible and consistent initial values.
    % \end{proposition}
    % The proof follows the usual arguments \cite{berberich2020data} and \cite{mayne2000constrained}, and is therefore omitted here.
    % Two applications of Corollary \ref{cor:Willems} to the Data-Enabled Predictive Control (DeePC) \cite{coulson2019data} of incomplete controllable descriptor systems can be found in the next section. 
    \vspace{-0.1\baselineskip}
\section{Applications to DeePC of descriptor systems}
\vspace{-0.1\baselineskip}
    %DeePC \cite{coulson2019data} replaces traditional model-based prediction with a data-driven optimization framework, becoming a more and more attracting control paradigm. Using a {\emph{fixed}} pre-collected input-output dataset \(({\bar u}_{[0, T-s]}, {\bar y}_{[0, T-s]})\) which meets certain persistent excitation requirements to parameterize all input-output trajectories based on Willems' fundamental lemma, it constructs predictive models without online system identification. Recently, Schmitz et al. \cite{schmitz2022willems} have extend DeePC to descriptor systems relying on the assumption of R-controllability (as well as R-observability) \cite{schmitz2022willems}.
   %  Unfortunately, for uncontrollable systems, this paradigm may not be valid. Although  
%Theorem \ref{thm:R-cont} has provided a sufficient condition to guarantee R-controllability, its conservatism is still unclear. Besides, the verification of Theorem \ref{thm:R-cont} may be computationally cumbersome. Below, inspired by \cite{yu2021controllability} and based on Corollary \ref{cor:Willems},  we propose an online DeePC paradigm applicable to uncontrollable descriptor systems.

 {DeePC \cite{coulson2019data} replaces traditional model-based prediction with a data-driven optimization framework, emerging as an increasingly attractive control paradigm. This approach constructs predictive models without online system identification by leveraging a \emph{fixed pre-collected input-output dataset} \(({\bar u}_{[0, T-s]}, {\bar y}_{[0, T-s]}) \)—satisfying persistent excitation requirements—to parameterize all input-output trajectories based on Willems' fundamental lemma. Recent work by Schmitz et al. \cite{schmitz2022willems} extended DeePC to descriptor systems under the assumptions of {R-controllability} and {R-observability}.
However, this paradigm faces limitations for {uncontrollable systems}. Although Theorem \ref{thm:R-cont} provides a sufficient condition for guaranteeing R-controllability, its {conservatism remains unclear}, and computational challenges arise in verifying the theorem for practical implementation. To address these limitations, we propose an {online DeePC paradigm applicable to uncontrollable descriptor systems}, inspired by \cite{yu2021controllability} and building on Corollary~\ref{cor:Willems}.}

    { 
Suppose system \eqref{2_1} is R-observable and regular, and let \(({u}_{[0, T-s]}, {y}_{[0, T-s]})\) be an input-output trajectory where \(u_{[0,T-1]}\) is persistently exciting of order \(N+L+\delta+s-1\).}  {Consider regulating the output \(y(k)\) of system \eqref{2_1} to the reference \(y^s \in \mathbb{R}^p\) using a receding horizon control strategy. For a prediction horizon \(L\), the control input sequence \(\hat{u} = (\hat{u}(t), \dots, \hat{u}(t+L-1))\) is obtained by solving the following optimization problem (only the first control input \(\hat{u}(t)\) is implemented before receding the horizon):
\begin{subequations}\label{6_1}
\begin{align}
\underset{\substack{(\hat{u},\hat{y}) \\ \alpha(t)}}{\min} \quad & 
\sum_{k=t}^{t+L-1} \left( \left\| \hat{y}(k) - y^{s} \right\|_{Q}^{2} + 
\left\| \hat{u}(k) \right\|_{R}^{2} \right) \label{eq:6_1a} \\
\text{s.t.} \quad & 
\begin{bmatrix}
u_{[t-N,t-1]} \\
\hat{u}_{[t,t+L-1]} \\
y_{[t-N,t-1]} \\
\hat{y}_{[t,t+L-1]}
\end{bmatrix} = 
\begin{bmatrix}
H_{N+L}({u}_{[0,T-s]}) \\
H_{N+L}({y}_{[0,T-s]})
\end{bmatrix} \alpha(t) \label{eq:6_1b} \\
& 
\begin{bmatrix}
\hat{u}_{[t+L-N,t+L-1]} \\
\hat{y}_{[t+L-N,t+L-1]}
\end{bmatrix} = 
\begin{bmatrix}
0 \\ \vdots \\ 0 \\
y^{s} \\ \vdots \\ y^{s}
\end{bmatrix} \label{eq:6_1c} \\
& \hat{u}(k) \in \mathbb{U}, \quad t \leq k \leq t+L-1 
\end{align}
\end{subequations}
where \((\hat{u},\hat{y}):[t, t+L-1] \to \mathbb{R}^{m} \times \mathbb{R}^{p}\), \(\alpha(t) \in \mathbb{R}^{T-N-L-s+2}\), \((T-s+1) \leq t\), and \((q+s-1) \leq N \leq t\). The set \(\mathbb{U} \subset \mathbb{R}^m\) defines feasible control inputs. The cost function weights \(Q \succcurlyeq 0 \in \mathbb{R}^{p \times p}\) and \(R \succcurlyeq 0 \in \mathbb{R}^{m \times m}\) (with \(\|u\|_R = u^\intercal R u\)) penalize tracking error and control effort respectively. {Notably, unlike DeePC in \cite{schmitz2022willems}, the initial trajectory \((u_{[0, T-s]}, y_{[0, T-s]})\) in \eqref{6_1} cannot be an arbitrarily generated offline trajectory.}}

 { Corollary \ref{cor:Willems} guarantees the online DeePC remains applicable even when the system (\ref{2_1}) is not R-controllable. Indeed, if the input-output subsequence $(\hat{u}_{[t, t+L-1]}, \hat{y}_{[t, t+L-1]})$ generated by system (\ref{2_1}) can be concatenated with historical data $(u_{[t-N, t-1]}, y_{[t-N, t-1]})$ to form a complete trajectory segment starting from the dynamically updated initial segment $(u_{[0,T-s]},y_{[0,T-s]})$, then this segment satisfies the constraint in (\ref{eq:6_1b}). On the other hand, any trajectory $(\hat{u}_{[t, t+L-1]}, \hat{y}_{[t, t+L-1]})$ satisfying (\ref{eq:6_1b}) corresponds to a trajectory of system (\ref{2_1}). This establishes that (\ref{6_1}) exactly characterizes the feasible solution space of the system’s true dynamics via data-driven parameterization. Thus, Corollary \ref{cor:Willems} essentially ensures that online DeePC, through dynamically updated real-time data, maintains a precise representation of the system’s behavior.}
    
    To verify the universality and effectiveness of the proposed theoretical method, this section presents two case studies on typical infrastructure systems: power systems and water networks.
     \subsection{Case Study $1$: Frequency Regulation in Power Networks}
     Consider a power system with $\mathcal{M} = \{1,\dots,n\}$ buses and $\mathcal{G}\subseteq\mathcal{M}$ generators. Without loss of generality we assume  $\mathcal{G} = \{1,\dots,g\}$. The grid parameters are described by the bus admittance matrix $Y = \bar{G}+i\bar{B} \in \mathbb{C}^{n\times n}$, where $\bar{G}\in \mathbb{R}^{n\times n}$ describes the conductances and $\bar{B}\in\mathbb{R}^{n\times n}$ describes the susceptances of all transmission lines.
    
    In this work, the following assumptions are adopted, consistent with those commonly used in the analysis of frequency control and power system dynamics \cite{simpson2013synchronization} \cite{schiffer2016survey} \cite{song2015small}:

     \begin{enumerate}
          \item The system is lossless ($\bar{G} = 0$).
          \item Generators are modeled as second-order synchronous generators.
          \item Loads are represented as constant power demands equal to zero.
          \item Bus voltages are assumed to be constant. 
          \item The voltage angle differences are small.
          \item The network topology remains connected.
     \end{enumerate} 

    Following \cite{bergen1981structure}, the dynamics of a synchronous generator $ i \in \mathcal{G} $ are described by
    \begin{equation} \label{power_sys}
    M_i \ddot{\theta}_i + D_i \dot{\theta}_i = p_i - \sum_{j \in \mathcal{M}} \bar{B}_{i,j} \sin(\theta_i - \theta_j),
    \end{equation}
    where $ M_i > 0 $ is the inertia of generator, $ D_i > 0 $ the damping coefficient, $ p_i \geq 0 $ the mechanical power input, and $ \theta_i \in \mathbb{R} $ the rotor phase angle.
    
    Rewriting (\ref{power_sys}) as a first-order system:
    \begin{align} \label{6_2}
    \begin{cases}
    \dot{\theta}_i = \omega_i, \\
    \dot{\omega}_i = M_i^{-1} \bigg( p_i - D_i \omega_i - \sum_{j \in \mathcal{M}} \bar{B}_{i,j} \sin(\theta_i - \theta_j) \bigg).
    \end{cases}
    % \quad \forall i \in \mathcal{G}.
    \end{align}
    
    For non-generator nodes ($ \forall i \in \mathcal{M} \setminus \mathcal{G} $), the algebraic power balance equations are:
    \begin{equation} \label{eq:algebraic}
     - \sum_{j \in \mathcal{M}} \bar{B}_{i,j} \sin(\theta_i - \theta_j) = 0,
    \end{equation}
    consistent with the constant power load assumption (cf. Assumption $3$).

    We employ the discrete-time linear descriptor system (LDS) model proposed by \cite{schmitz2022data}, which linearizes the power flow equations under small voltage angle differences (Assumption $5$). The system matrices are defined as \( E, A \in \mathbb{R}^{(n+g)\times(n+g)} \), \( B \in \mathbb{R}^{(n+g)\times g} \), and \( C \in \mathbb{R}^{p\times(n+g)} \). Let \( L \in \mathbb{R}^{n\times n} \) denote the graph Laplacian matrix associated with the susceptance matrix \( \bar{B} \) of the power network, given by:  
    \[
    L = \operatorname{diag}\left( \sum_{j \in \mathcal{M}} \bar{B}_{1,j}, \ldots, \sum_{j \in \mathcal{M}} \bar{B}_{n,j} \right) - \bar{B}.
    \]  
    As the study focuses on discrete-time descriptor systems, we employ the forward Euler method to discretize the continuous power flow dynamics, with $\tau = 0.1$ defining the discretization step size. The system matrices are structured as:  
    \[
    E = \begin{bmatrix} 
    0 & I_g & 0 \\ 
    \tilde{M} & 0 & 0 \\ 
    0 & 0 & 0 
    \end{bmatrix}, \quad
    A = E - \tau \left[ 
    \begin{array}{c c c} 
    -I_{g} & 0 & 0 \\ 
    \tilde{D} & \multicolumn{2}{c}{\multirow{2}{*}{\Large $L$}} \\ 
    0 & \multicolumn{2}{c}{} \\
    \end{array} 
    \right], \quad
    B = \tau \begin{bmatrix} 
    0 \\ 
    I_g \\ 
    0 
    \end{bmatrix},
    \]  
    where \( \tau > 0 \) represents the discretization time step, and 
    \[
    \tilde{M} = \operatorname{diag}(M_1, \ldots, M_g), \quad \tilde{D} = \operatorname{diag}(D_1, \ldots, D_g),
    \]  
    denote the diagonal matrices of generator inertias and damping coefficients, respectively.  
    
    The state vector \( x(t) \in \mathbb{R}^{n+g} \) and control input \( u(t) \in \mathbb{R}^g \) are defined as:  
    \[
    x(t) = \begin{bmatrix} 
    \omega_{\mathcal{G}}(t) \\ 
    \theta_{\mathcal{G}}(t) \\ 
    \theta_{\mathcal{M} \setminus \mathcal{G}}(t) 
    \end{bmatrix}, \quad
    u(t) = p(t),
    \]  
    where \( \omega_{\mathcal{G}} \) and \( \theta_{\mathcal{G}} \) represent the angular frequencies and rotor angles of generators, respectively, while \( \theta_{\mathcal{M} \setminus \mathcal{G}} \) denotes the voltage angles at non-generator buses.

    However, in order to verify Corollary \ref{cor:Willems}, we need to disrupt the R-controllable structure of the system. Define the input matrix as
    \[ B = \tau\begin{bmatrix} \mathbf{0}_{g \times g} \\ \begin{bmatrix} 1 & 1 & 0 \\ -1 & 0 & 1 \\ 0 & -1 & -1 \end{bmatrix} \\ \mathbf{0}_{(n-g) \times g} \end{bmatrix}, \]  
    which implies that the input \( u \) influences the mechanical power of each generator through a linearly dependent combination. We consider a nine-bus system \cite{schulz1977long}, with parameters from \cite{scholtz2004observer}. The IEEE 9-bus power system is shown in Fig. \ref{fig:nine bus power system}, including three generators $(g_1, g_2, g_3)$ and three consumers $(LumpA, LumpB, LumpC)$.
    \begin{figure}[t]
    \centering    \includegraphics{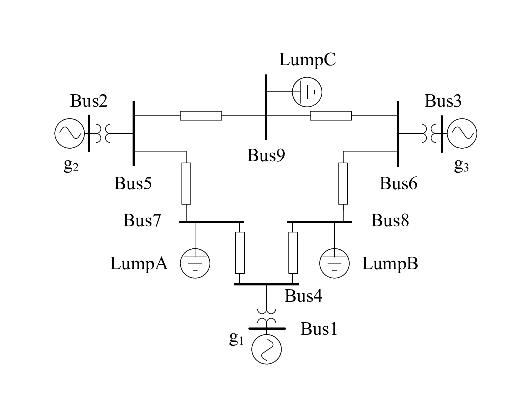} 
    %% Use \caption command for figure caption and label.
    \caption{IEEE 9-bus power system}\label{fig:nine bus power system}
    \end{figure}

    In Case Study 1, the controllable state $x_1$ and the uncontrollable state $x_7$ were selected as system outputs, which ensures that system (\ref{2_1}) is R-observable. The system parameters were configured as \( N = 6 \), \( L = 14 \), \( T = 100 \), \( K = 200 \) and generate online input-output trajectories \( (u_{[0,T-s]}, y_{[0,T-s]}) \), where the input sequence \( \{u_t\}_{t=0}^{T-1} \) was independently and uniformly sampled from the interval \( [-1, 1] \), guaranteeing persistent excitation of order \( N + L + \delta + s - 1 \). At time $t = T-s+1, T-s+2, \dots,K$, given $(u_{[0,t-1]},y_{[t-1]})$, the control input \( \hat{u}_t \) is determined by solving the optimization problem (\ref{6_1}) with \( Q=5I_2 \), \( R=I_3 \) and \( \mathbb{U}=[-1,1] \) for all $t\in\mathbb{N}$.
    % A fixed pre-collected dataset \( (\bar{u}_{[0,T-s]}, \bar{y}_{[0,T-s]}) \) was constructed, where the input sequence \( \{\bar{u}_t\}_{t=0}^{T-1} \) was independently and uniformly sampled from the interval \( [-1, 1] \), guaranteeing persistent excitation of order \( N + L + \delta + s - 1 \). 
    In Figure. \ref{fig:power_system_trajectories}, we illustrate steering the system's outputs to the setpoints $(y^{s,1} = \begin{bmatrix}
	   3 & 3 
	\end{bmatrix}^{\top})$ and $( y^{s,2} = \begin{bmatrix}
		1 & 5
	\end{bmatrix}^{\top})$ one by one using using online DeePC. It turns out that the controllable component $x_1$ is successfully steered to the equilibrium, while the uncontrollable component $x_7$ is not. 
    \begin{figure}[t]
	\centering
        \includegraphics[width=0.5\textwidth,keepaspectratio]{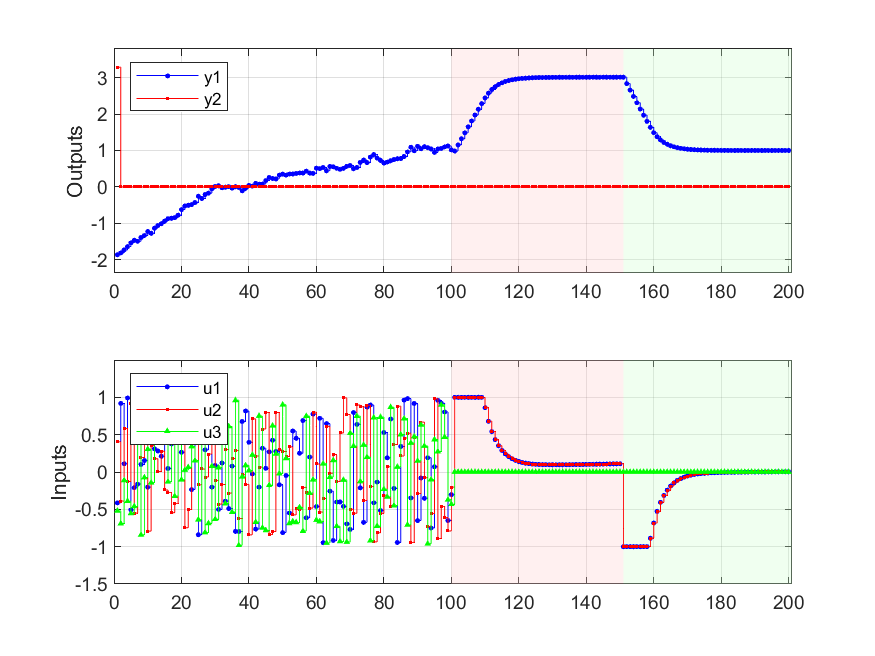}
	\caption{Output trajectory of the power network for Online DeePC. The red-shaded region indicates the transition period from random excitation inputs to the optimal control input, successfully steering the system's controllable output $x_1$ to equilibrium $y^{s,1}$, and subsequent tracking of the new equilibrium $y^{s,2}$ begins in the green-shaded phase starting at $t = 151$. Meanwhile, the output $x_7$ is uncontrollable. }
	\label{fig:power_system_trajectories}
    \end{figure}
 \subsection{Case Study 2: Pressure Control in Water Networks}   
     Following prior hydraulic network modeling approaches \cite{burgschweiger2009optimization} and \cite{boulos2006comprehensive}, water networks can be represented as directed graphs \( \mathcal{G} = (\mathcal{V},\mathcal{E}) \), where the vertex set \( \mathcal{V} \) comprises three functional components: (R) reservoirs, (J) junctions, and (T) storage tanks. The edge set \( \mathcal{E} \) consists of hydraulic elements, including pipes, pumps, and valves. The key variables are the pressure head $h_i$ at each network node $i$ and the $Q_{ij}$ from node $i$ to $j$.

    The hydraulic model governing the network dynamics includes 
    constant reservoir heads, flow balance equations at junctions 
    and tanks, and pressure difference equations along all edges:
    \begin{enumerate}
        \item Reservoirs maintain fixed hydraulic heads: 
        \[
         h_i^{\text{R}} = \text{constant}, \quad \forall i \in \mathcal{V}_{\text{R}}.
         \]
        \item Junctions obey flow continuity:  
         \[
         d_i = \sum_{j \in \mathcal{N}_i^-} Q_{ji} - \sum_{k \in \mathcal{N}_i^+} Q_{ik}, \quad \forall i \in \mathcal{V}_{\text{J}}
         \]  
         where \( d_i \) denotes water demand and \( \mathcal{N}_i^- \), \( \mathcal{N}_i^+ \) represent upstream/downstream neighbors.
         \item Storage tanks follow mass balance:  
         \[
         A_i \frac{dh_i}{dt} = \sum_{j \in \mathcal{N}_i^-} Q_{ji} - \sum_{k \in \mathcal{N}_i^+} Q_{ik}, \quad \forall i \in \mathcal{V}_{\text{T}}
         \]  
         with \( A_i \) being the tank's cross-sectional area.
         \item Pipes exhibit pressure-dependent flow:  
         \[
         Q_{ij} = f_{ij}(h_i - h_j), \quad \forall (i,j) \in \mathcal{E}_{\text{pipe}}.
         \]  
         \item Pumps impose fixed head gain:  
         \[
         h_j - h_i = +\Delta h_{ij}^{\text{pump}}, \quad \forall (i,j) \in \mathcal{E}_{\text{pump}}.
         \]  
         \item Valves introduce prescribed head loss:  
         \[
         h_j - h_i = -\Delta h_{ij}^{\text{valve}}, \quad \forall (i,j) \in \mathcal{E}_{\text{valve}}.
         \]
    \end{enumerate}

    The hydraulic relationship in pipeline systems adopts a generalized functional form \( f_{ij}(\cdot) \), allowing substitution with domain-specific formulations such as the Darcy-Weisbach equation.

    Consider the EPANET 3 \cite{rossman2000epanet} linearized in steady state with non-zero pressure drops. The topological configuration of the EPANET 3 network is illustrated in Fig. \ref{fig:epanet3}. We employ the water network descriptor formulation developed by \cite{pasqualetti2014control}.  we identify $x_1$, $x_2$, $x_3$, and $x_4$ denote, respectively, the pressure at the reservoir $R_2$, at the reservoir $R_1$ and at the tanks $T_1$, $T_2$, and $T_3$, at the junction $P_2$, and at the remaining junctions. The descriptor model for the EPANET 3 network is
    \begin{equation}
        \begin{aligned}
        {\left[\begin{array}{c}
        \dot{x}_1(t) \\
        M \dot{x}_2(t) \\
        0 \\
        0
        \end{array}\right] } & =\left[\begin{array}{cccc}
        0 & 0 & 0 & 0 \\
        0 & A_{22} & 0 & A_{24} \\
        A_{31} & 0 & A_{33} & A_{34} \\
        0 & A_{42} & A_{43} & A_{44}
        \end{array}\right]\left[\begin{array}{l}
        x_1(t) \\
        x_2(t) \\
        x_3(t) \\
        x_4(t)
        \end{array}\right]\\
        y & =\left[\begin{array}{llll}
        C_1 & C_2 & C_3 & C_4
        \end{array}\right] x,
        \end{aligned}
    \end{equation}
    where the pattern of zeros is due to the network interconnection structure, and $M = diag(1, A_1, A_2, A_3)$ corresponds to the dynamics of the reservoir $R_1$ and the tanks $T_1$, $T_2$, and $T_3$. Consistent with the methodology adopted in Case Study 1, we employ Euler-forward discretization with the discretization step size $\tau = 0.1$.

\begin{figure}[t]
    \centering
\includegraphics[width=0.5\textwidth,keepaspectratio]{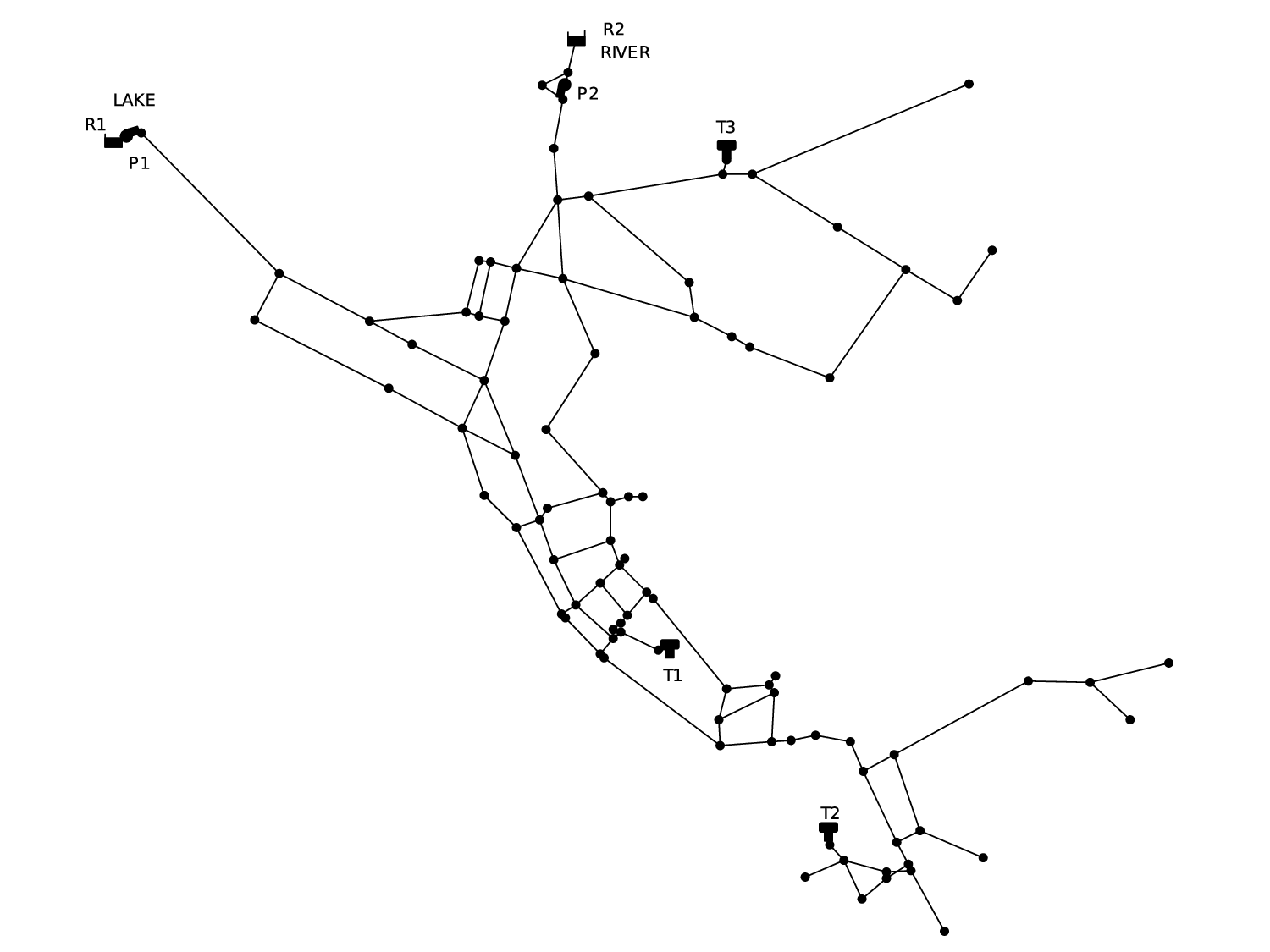}
    \caption{The structure of the EPANET 3}
    \label{fig:epanet3}
    \end{figure}

    To ensure the descriptor system maintains regularity and violates R-controllability, we strategically allocate the control input to the first junction of $x_3$. In order to prove the correctness of Corollary \ref{cor:Willems}, we chose the pressure at reservoir $R_1$ and tank $T_3$ as the output, hereby guaranteeing the system (\ref{2_1}) is R-observable.

    In the present simulation framework, the three-dimensional vector $u$ is designated as the input of the system, with its components $u = \begin{bmatrix}
        u_1 & u_2 & u_3
    \end{bmatrix}^{\top}$. 
    Case Study 2 adopts the parameters $N=5, L=12, T=100, \,\ \text{and}\,\ K=200$. For all \( t \in \mathbb{N} \), the control input \( \hat{u}_t \) is determined by solving the optimization problem (\ref{6_1}) with \( Q=5I_2 \), \( R=I_3 \) and \( \mathbb{U}=[-1,1] \). In Figure \ref{fig:water_network_trajectories}, we illustrate simulations of steering the system's outputs to the setpoints $(y^{s,1} = \begin{bmatrix}
	   5 & 3 
	\end{bmatrix}^{\top})$ and $( y^{s,2} = \begin{bmatrix}
		3 & 5
	\end{bmatrix}^{\top})$ one by one. It turns out that both case studies confirm that online DeePC maintains accurate tracking of the controllable output \( y(t) \) to its reference \( y^s(t) \), despite the presence of the uncontrollable output.
    \begin{figure}[t]
	\centering	\includegraphics[width=0.5\textwidth,keepaspectratio]{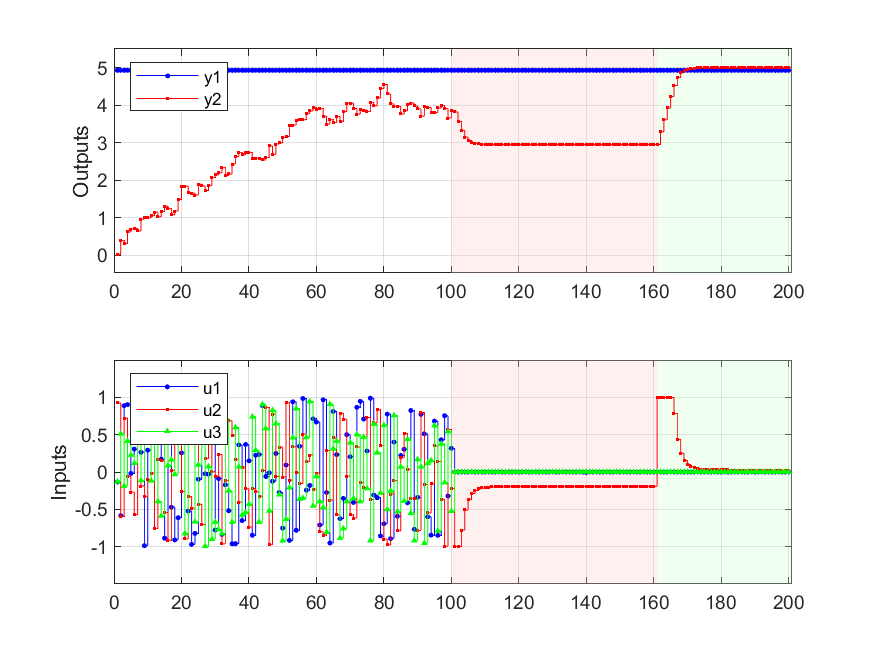} 
	\caption{Output trajectory of the water network via Online DeePC. At time $t = 101$, the transition from a random input signal to the optimal control input (red shaded), which steers the system's controllable output into the equilibrium $y^{s,1}$, can be seen. At time $t = 161$ (green shaded), the desired equilibrium is changed to $y^{s,2}$.}
	\label{fig:water_network_trajectories}
    \end{figure}

\section{Conclusion}
    This paper establishes a theoretical data-driven framework for analyzing and controlling discrete-time descriptor systems, eliminating the need for explicit state-space representations. The main contributions include i) the development of sufficient conditions for controllability (R-controllability and C-controllability) and observability (R-observability and C-observability) through the joint construction of forward and backward data matrices, which resolves the challenge of non-causality in descriptor systems, and ii) an extension of Willems’ fundamental lemma by relaxing the controllability and persistency of excitation assumptions, leading to a paradigm of online DeePC for uncontrollable descriptor systems. Simulations in power and water networks confirm that online DeePC enables controllable outputs to accurately track the reference despite the presence of uncontrollable outputs, demonstrating the theoretical efficacy of the proposed method in partially uncontrollable systems.
%% Refer following link for more details.
%% https://en.wikibooks.org/wiki/LaTeX/Mathematics
%% https://en.wikibooks.org/wiki/LaTeX/Advanced_Mathematics

%% Use a table environment to create tables.
%% Refer following link for more details.
%% https://en.wikibooks.org/wiki/LaTeX/Tables

%% Use figure environment to create figures
%% Refer following link for more details.
%% https://en.wikibooks.org/wiki/LaTeX/Floats,_Figures_and_Captions

%\appendix
\section*{Appendix: Proof of Theorem \ref{thm:Willems}}
\label{app1}
    We start by proving the first statement using a double-inclusion argument.		
	
    It is obviously that 
	\begin{align*}
		z_2(k)=-\sum_{m=0}^{s-1}N^mB_2u(k+m) \in \mathcal{N}.
	\end{align*}	
    Using an argument similar to (\ref{5_1}), one can directly show that the left-hand side of (\ref{5_3}) is included in its right hand side. 
	
    To show the other direction, we show that the left kernel of the matrix
	\begin{equation}
		\begin{bmatrix}
			H_1(\bar{z}_{[0,T-L-s+1]})\\
			H_{L+s-1}(\bar{u}_{[0,T-1]})
		\end{bmatrix}
		\label{5_5}
	\end{equation}
	is orthogonal to $(\cal{R}+\cal{K}\begin{bmatrix}
		z_{10}
	\end{bmatrix})\times \mathcal{N}\times \mathbb{R}^{\text{m(L+s-1)}}$. To this end, let
	\begin{equation}
		v^{\top}=\begin{bmatrix}
			\zeta^{\top} & \beta^{\top} & \eta_1^{\top} & \dots & \eta_{L+s-1}^{\top}
		\end{bmatrix}
		\label{5_6}
	\end{equation}
    be an arbitrary row vector in the left kernel of the matrix (\ref{5_5}), where $\zeta \in \mathbb{R}^{q}$,$\beta \in \mathbb{R}^r$,$\eta_1,\eta_2,\dots,\eta_{L+s-1}\in \mathbb{R}^{m}$. Since $\delta \ge \delta_{\text{min}}$, we know there exists $\alpha_{0k},\alpha_{1k},\dots,\alpha_{\delta-1,k} \in \mathbb{R}$ such that
	\begin{equation}
		A_1^k+\sum_{j=0}^{\delta-1}\alpha_{jk}A_1^j=0_{q\times q},\forall k=\delta,\delta+1,\dots.
		\label{5_7}
	\end{equation}
    The above equation implies that $A_1^kB_1=-\sum_{j=0}^{\delta-1}\alpha_{jk}A_1^jB_1$ and $A_1^kz_{10}=-\sum_{j=0}^{\delta-1}\alpha_{jk}A_1^jz_{10}$ for all $k = \delta, \delta+1,\dots$. Therefore, in order to show the left kernel of the matrix (\ref{5_5}) is orthogonal to $(\cal{R}+\cal{K}\begin{bmatrix}
		z_{10}
	\end{bmatrix})\times \mathcal{N} \times \mathbb{R}^{\text{m(L+s-1)}}$, it suffices to show the following
	\begin{subequations} \label{plant}
		\begin{align}			\eta_1^{\top}=\eta_2^{\top}=\dots=\eta_{L+s-1}^{\top}=0_m^{\top},\label{5_8a}\\	\zeta^{\top}B_1=\zeta^{\top}A_1B_1=\dots=\zeta^{\top}A_1^{\delta-1}B_1=0_m^{\top},\label{5_8b}\\			\zeta^{\top}z_{10}=\zeta^{\top}A_1z_{10}=\dots=\zeta^{\top}A_1^{\delta-1}z_{10}=0,\label{5_8c}\\
			\beta^{\top}B_2 =\beta^{\top}NB_2=\dots= \beta^{\top}N^{s-1}B_2=0_m^{\top}.\label{5_8d}
		\end{align}
	\end{subequations}
    In order to prove (\ref{5_8a}), we let $\omega_0,\omega_1,\dots,\omega_{\delta}\in\mathbb{R}^{n+m(\delta+L+s-1)}$ be such that $\omega_0=\begin{bmatrix}
			\beta^{\top}& \zeta^{\top}& \eta_1^{\top}&\dots & \eta_{L+s-1}^{\top} & 0_{m\delta}^{\top}
	\end{bmatrix}^{\top}$ and $w_j$ equal\\ $\begin{bmatrix}
        \beta^{\top} & \zeta^{\top}A_1^j & \zeta^{\top}A_1^{j-1}B_1 & \dots & \zeta^{\top}B_1 & \eta_1^{\top} & \dots & \eta_{L+s-1}^{\top} & 0_{m(\delta-j)}^{\top}
	\end{bmatrix}^{\top}$, for $j=1,\dots,\delta$. Since $v^{\top}$ is the left kernel of matrix (\ref{5_5}), using (\ref{2_1}) one can verify that $\omega_0^{\top},\omega_1^{\top},\dots,\omega_{\delta}^{\top}$ are the left kernel of
	\begin{equation}
		\begin{bmatrix}
			H_1(\bar{z}_{2[0,T-L-s-\delta+1]})\\
			H_1(\bar{z}_{1[0,T-L-s-\delta+1]})\\
			H_{L+s+\delta-1}(\bar{u}_{[0,T-1]})
		\end{bmatrix}.
		\label{5_9}
	\end{equation}
	Let $\alpha_{\delta\delta}=1$ and $k=\delta$ in (\ref{5_7}), we have $\sum_{j=0}^{\delta}\alpha_{jk}A_1^j=0_{q\times q}$. Hence
	\begin{equation}
		\begin{aligned}
			\sum_{j = 0}^{\delta}\alpha_{j\delta}\omega_j^{\top}
			&=\begin{bmatrix}
				\sum_{j = 0}^{\delta}\alpha_{j\delta}\beta^{\top} & \sum_{j = 0}^{\delta}\alpha_{j\delta}\zeta^{\top}A_1^j & \gamma^{\top}
			\end{bmatrix}\\
			&=\begin{bmatrix}
				\sum_{j = 0}^{\delta}\alpha_{j\delta}\beta^{\top} & 0_q & \gamma^{\top}
			\end{bmatrix},
		\end{aligned}
		\label{5_10}
	\end{equation}
	for some vector $\gamma \in \mathbb{R}^{m(L+s+\delta-1)}$. Since row vectors $\omega_0^{\top},\omega_1^{\top},\dots,\omega_{\delta}^{\top}$ are the left kernel of matrix (\ref{5_9}), equation (\ref{5_10}) implies that $\gamma^{\top}$ is the left kernel of matrix
	\begin{equation}
		H_{L+s+\delta-1}(u_{[0,T-1]}).
		\label{5_11}
	\end{equation}
    Since $u_{[0,T-1]}$ is collectively persistently exciting of order $(L+\delta+s-1)$, matrix (\ref{5_11}) has full row rank. Therefore
	\begin{equation}
		\gamma = 0_{m(L+\delta+s-1)}.
		\label{5_12}
	\end{equation}
	Observe that the last $m$ entries of $\gamma$ are $\alpha_{\delta\delta}\eta_{L+s-1}=\eta_{L+s-1}$, hence equation (\ref{5_12}) implies that $\eta_{L+s-1}=0_{m}$. Then the last $2m$ entries of $\gamma$ are given by $\begin{bmatrix}
        \alpha_{\delta\delta}\eta_{L+s-2}^{\top}+\alpha_{(\delta-1)\delta}\eta_{L+s-1}^{\top} & \alpha_{\delta\delta}\eta_{L+s-1}^{\top}
	\end{bmatrix}$. Since $\alpha_{\delta\delta}=1$ and $\eta_{L+s-1}=0_{m}$ equation(\ref{5_12}) also implies that $\eta_{L+s-2}=0_{m}$. By repeating similar induction we can prove that (\ref{5_8a}) holds.	
    
    Next, since (\ref{5_8a}) holds, the first $m\delta$ entries in $\gamma$ are
	$$
        \begin{bmatrix}			          \sum_{j=1}^{\delta}\alpha_{j\delta}\zeta^{\top}A_1^{j-1}B_1 & \sum_{j=2}^{\delta}\alpha_{j\delta}\zeta^{\top}A_1^{j-2}B_1 & \dots & \alpha_{\delta\delta}\zeta^{\top}B_1
	\end{bmatrix}^{\top}.
	$$
    By combining this with (\ref{5_12}) we get that
	\begin{equation}
		0_m^{\top} = \sum_{j=k}^{\delta}\alpha_{j\delta}\zeta^{\top}A_1^{j-k}B_1,\forall k=1,\dots,\delta.
		\label{5_13}
	\end{equation}
    Since $\alpha_{\delta\delta}=1$, considering $k=\delta$ in (\ref{5_13}) implies that $\zeta^{\top}B_1=0_m$. Substitute this back into (\ref{5_12}) and considering $k=\delta-1$ implies that $\zeta^{\top}A_1B_1=0_m$. By repeating a similar reasoning for $k=\delta-2,\delta-3,\dots,1$ we can prove that (\ref{5_8b}) holds.
	
    Further, by using (\ref{5_7}) and (\ref{5_8b}) we can show that $\zeta^{\top}A_1^kB=0_m$ for all $k\ge1$. Combining this together with the fact that row vector (\ref{5_6}) is the left kernel of matrix (\ref{5_5}) and that (\ref{5_8b}) also holds, it follows that,
	\begin{align*}
		0&=\zeta^{\top}z_{1k}\\
		&=\zeta^{\top}(A_1^kz_{10}+\sum_{j=0}^{k-1}A_1^{k-j-1}B_1u_j)\\
		&=\zeta^{\top}A_1^kz_{10},\forall k=0,\dots,T-L-s+1.
	\end{align*}
    Since $T\ge\delta+L+s-1$ by assumption, we conclude that (\ref{5_8c}) holds.	
	
    Next, we will prove the (\ref{5_8d}). As $\omega_0$ is the left kernel of (\ref{5_9}), we can get $\beta^{\top}H_1(\bar{z}_{2[0,T-L-s-\delta+1]})$. It means that $$\beta^{\top} \begin{bmatrix}
    \bar{z}_{2}(0) & \bar{z}_{2}(1) & \dots & \bar{z}_{2}(T-L-s-\delta+1)
	\end{bmatrix}=0,$$
	$$
	-\beta^{\top} \begin{bmatrix}
			B_2 & NB_2 & N^2B_2 & \dots & N^{s-1}B_2
	\end{bmatrix}H_s(u_{[0,T-L-\delta]})
	=0,
	$$
	$$
	-\begin{bmatrix}
			\beta^{\top}B_2 & \beta^{\top}NB_2 & \beta^{\top}N^2B_2 & \dots & \beta^{\top}N^{s-1}B_2
	\end{bmatrix}H_{s}(u_{[0,T-L-\delta]})=0.
	$$		
    It can be seen that $H_{s}(u_{[0,T-L-\delta]})$ is constructed from the first $s$ rows of $H_{L+s+\delta-1}(u_{[0,T-1]})$, so $H_{s}(u_{[0,T-L-\delta]})$ is also full row rank. Thus, $\beta^{\top}B_2 =\beta^{\top}NB_2=\dots= \beta^{\top}N^{s-1}B_2=0$, (\ref{5_8d}) holds.
    
    We now prove the second statement. 	
	
    If $\{u_{[0,L-1]},z_{{[0,L-1]}},y_{[0,L-1]}\}$ is an input-state-output trajectory generated by (\ref{2_1}),	
	\begin{align*}
		\begin{bmatrix}
			u_{[0,L-1]}\\
			y_{[0,L-1]}
		\end{bmatrix}
		=\left[\begin{array}{c:c}
			0 & [I_{mL}\,\, 0]\\
			O_{sL} & T_L
		\end{array}\right]
		\begin{bmatrix}
			z_1(0)\\
			u_{[0,L+s-2]}
		\end{bmatrix}.
	\end{align*}	
    First, the necessity of the second statement is demonstrated. Given the input, state, and output trajectories, suppose (\ref{5_2}) holds.
    Let
	\begin{equation}
		\begin{aligned}
			\bar{z}_0=
			\begin{bmatrix}
				\bar{z}_{10}\\
				\bar{z}_{20}
			\end{bmatrix}=\begin{bmatrix}
				H_1(z_{1[0,T - L - s + 1]})\\
				H_1(z_{2[0,T - L - s + 1]})
			\end{bmatrix}g, g\in\mathbb{R}^{T - L - s + 2},\\
			\{
			\begin{array}{l}
				\bar{z}_1(t + 1)=A_1\bar{z}_1(t)+B_1\bar{u}(t)\\
				N\bar{z}_2(t + 1)=\bar{z}_2(t)+B_2\bar{u}(t)
			\end{array}, 0\le t \le L - 2.
		\end{aligned}
		\label{5_14}
	\end{equation}	
    Then from (\ref{5_3}) we know $\bar{z}_{0}$ satisfies (\ref{5_4}), and one can  verify that $(\bar{u}_{[0,L-1]},\bar{z}_{[0,L-1]},\bar{y}_{[0,L-1]})$ is indeed an input-state-output trajectory of system (\ref{2_1}).	
	
    Conversely, let $\bar{u}_{[0,L-1]},\bar{z}_{[0,L-1]}\bar{y}_{[0,L-1]}$ be an input-state-output trajectory of system (1) with $\bar{z}_{0}\in (\cal{R}+\cal{O}+\cal{K}\begin{bmatrix}
		z_{10}
	\end{bmatrix})\times \mathcal{N}$. Then there exists
	\begin{equation}
		\bar{z}_{10}^a\in\mathcal{R}+\mathcal{K}[z_{10}],\quad
		\bar{z}_{10}^{b}\in\mathcal{O},\quad
		\bar{z}_{20}\in\mathcal{N}
		\label{5_15}
	\end{equation}
    such that  $\bar{z}_{10}=\bar{z}_{10}^a+\bar{z}_{10}^b$ and
	\begin{equation}
		\begin{bmatrix}
			\bar{u}_{[0,L-1]}\\
			\bar{y}_{[0,L-1]}
		\end{bmatrix}=				\left[\begin{array}{c:c}
			0 & [I_{mL}\,\, 0]\\
			O_{sL} & T_L
		\end{array}\right]
		\begin{bmatrix}
			\bar{z}_{10}^a+	\bar{z}_{10}^b\\
			\bar{u}_{[0,L+s-2]}
		\end{bmatrix}.
		\label{5_16}
	\end{equation}	
    Further, using Cayley-Hamilton theorem, one can show that ${\cal O}\subset Ker(O_{sL})$ for any $L\in\mathbb{N}$. Hence (\ref{5_15}) implies
	\begin{equation}
		O_{sL}\bar{z}_{10}^b=0_{pL}.   
		\label{5_17}
	\end{equation}	
	Since $\bar{z}_{10}^a\in\mathcal{R}+\mathcal{K}[z_{10}]$ and $\bar{z}_{20} \in \mathcal{N}$, the first statement of Theorem \ref{thm:Willems} implies that there exists $g\in\mathbb{R}^{T-L-s+2}$ such that
	\begin{equation}
		\begin{bmatrix}
			\bar{z}_{10}^a\\
			\bar{u}_{[0,L+s-2]}
		\end{bmatrix}=
		\begin{bmatrix}
			H_1(z_{1[0,T-L-s+1]})\\
			H_{L+s-1}(u_{[0,T-1]})
		\end{bmatrix}g.
		\label{5_18}
	\end{equation}    
	Notice that
            \begin{equation}
		\left[\begin{array}{c:c}
			0 & [I_{mL}\,\, 0]\\
			O_{sL} & T_L    
		\end{array}\right]
		\begin{bmatrix}
			H_1(z_{1[0,T-L-s+1]})\\
			H_{L+s-1}(u_{[0,T-1]})
		\end{bmatrix}=
		\begin{bmatrix}
			H_L(u_{[0,T-s]})\\
			H_L(y_{[0,T-s]})\\	
		\end{bmatrix}.
		\label{5_19}
	\end{equation}    
         Substituting (\ref{5_17}), (\ref{5_18}) and (\ref{5_19}) into (\ref{5_16}) completes the proof.

%%       \cite{<label>} ==> [1]

%%
% Example citation, See \cite{lamport94}.

 \bibliographystyle{elsarticle-num} 
 \bibliography{example.bib}




% %% For numbered reference style

%   \textit{\LaTeX: a document preparation system},
%   Addison Wesley, Massachusetts,

% \end{thebibliography}
\end{document}